\documentclass[10pt,onecolumn,twoside]{IEEEtran}
\usepackage{amsmath}
\usepackage{amsthm}
\usepackage{amsfonts}
\usepackage{amssymb}
\usepackage{mathrsfs}
\usepackage{graphicx}
\usepackage{cite}

\newtheorem{corollary}{\bf Corollary}
\newtheorem{theorem}{\bf Theorem}
\theoremstyle{remark}

\newcommand{\norm}[1]{\left|\left|#1\right|\right|}

\usepackage{cases}
\usepackage{float}

\title{Zero Attracting PNLMS Algorithm and Its Convergence in Mean} 

\author{Rajib Lochan Das, ~\IEEEmembership{Member,~IEEE}, and Mrityunjoy
Chakraborty, ~\IEEEmembership{Senior Member,~IEEE}
\thanks{R. L. Das is with the Department of Electronics and Communication Engineering, Dr. B. C. Roy Engineering College, Durgapur, INDIA.
M. Chakraborty is with the Department of Electronics and Electrical Communication
Engineering,
Indian Institute of Technology, Kharagpur, INDIA.
(e.mail : rajib.das.iit@gmail.com;
 mrityun@ece.iitkgp.ernet.in).}}

\begin{document}

\maketitle
\thispagestyle{empty}
\pagestyle{empty}

\begin{abstract}  The proportionate normalized least mean square
(PNLMS) algorithm and its variants are by far the most popular
adaptive filters that are used to identify sparse systems. The
convergence speed of the PNLMS algorithm, though very high
initially, however, slows down at a later stage, even becoming
worse than sparsity agnostic adaptive filters like the NLMS. In
this paper, we address this problem by introducing a carefully
constructed $l_1$ norm (of the coefficients) penalty in the PNLMS
cost function which favors sparsity. This results in certain “zero
attractor” terms in the PNLMS weight update equation which
help in the shrinkage of the coefficients, especially the inactive
taps, thereby arresting the slowing down of convergence and also
producing lesser steady state excess mean square error (EMSE). We also carry
out the convergence analysis (in mean) of the proposed algorithm.


\end{abstract}


\section{Introduction}
\IEEEPARstart{I}{n}
real life, there exist many examples of systems that have
a sparse impulse response, having a few significant non-zero
elements (called active taps) amidst several zero or insignificant
elements (called inactive taps). One example of such systems is the network echo
canceller \cite{Radecki}-\cite{Krishna}, which uses both
packet-switched and circuit-switched components and has a total
echo response of about 64-128 ms duration out of which the
``active" region spans a duration of only 8-12 ms, while the
remaining ``inactive" part accounts for bulk delay due to network
loading, encoding and jitter buffer delays. Another example is the
acoustic echo generated due to coupling between microphone and
loudspeaker in hands free mobile telephony, where the sparsity of
the acoustic channel impulse response varies with the
loudspeaker-microphone distance \cite{Hansler}. Other well known
examples of sparse systems include HDTV where clusters of dominant echoes arrive
after long periods of silence \cite{Schreiber}, wireless multipath
channels which, on most of the occasions, have only a few clusters of significant paths \cite{Bajwa},
and underwater acoustic channels where the various multipath components caused by reflections
off the sea surface and sea bed have long intermediate delays \cite{Kocic}.
The last decade witnessed a flurry of research activities \cite{RDas}
that sought to develop sparsity aware adaptive filters which can
exploit the a priori knowledge of the sparseness of the system
and thus enjoy improved identification performance. The first
and foremost in this category is the proportionate normalized
LMS (PNLMS) algorithm \cite{PNLMS} which achieves faster initial
convergence by deploying different step sizes for different
weights, with each one made proportional to the magnitude
of the corresponding weight estimate. The convergence rate
of the PNLMS algorithm, however, slows down at a later
stage of the iteration and becomes even worse than a sparsity
agnostic algorithm like the NLMS \cite{hay}. This problem was
later addressed in several of its variants like the improved
PNLMS (i.e. IPNLMS) algorithm \cite{IPNLMS}, composite proportionate
and normalized LMS (i.e. CPNLMS) algorithm [10] and mu
law PNLMS (i.e. MPNLMS) algorithm \cite{Deng}. These algorithms
improve the transient response (i.e. convergence speed) of the
PNLMS algorithm for identifying sparse systems. However,
all of them yield almost same steady-state excess mean square
error (EMSE) performance as produced by the PNLMS. The
need to improve both transient and steady-state performance
subsequently led to several variable step-size (VSS), proportionate
type algorithms \cite{Liu_vss}-\cite{Paleologu_vss}.

In this paper, drawing ideas from \cite{Chen}-\cite{Hero}, we aim to improve
the performance of the PNLMS algorithm further, by introducing
 a carefully constructed $l_1$ norm (of the coefficients)
penalty in the PNLMS cost function which favors
sparsity\footnote{Some preliminary results of this work  were
earlier presented by the authors at
ISCAS 2014 \cite{Rajib_vss_ISCAS}.}. This results in a modified
PNLMS update equation with a “zero attractor” for all the
taps, named as the Zero-Attracting PNLMS (ZA-PNLMS) algorithm. The
zero attractors help in the shrinkage of the coefficients which is
particularly desirable for the inactive taps, thereby giving rise
to lesser steady state EMSE for sparse systems. Further, by
drawing the inactive taps towards zero, the zero attractors help
in arresting the sluggishness of the convergence of the PNLMS
algorithm that takes place at a later stage of the iteration,
caused by the diminishing effective step sizes of the inactive
taps. We show this by presenting a detailed convergence analysis
of the proposed algorithm, which is, however, a very daunting
task, especially due to the presence of a so-called gain matrix
and also the zero attractors in the update equation. To overcome
the challenges posed by them, we deploy a transform domain
equivalent model of the proposed algorithm and separately, an
elegant scheme of angular discretization of continuous valued
random vectors proposed earlier in \cite{Slock}.

\section{Proposed Algorithm}
Consider a PNLMS based adaptive
filter that takes $x(n)$ as input and updates a $L$ tap
coefficient vector\\
 ${\bf
w}(n)=[w_{0}(n),w_{1}(n),\cdots,w_{L-1}(n)]^T$ as \cite{PNLMS},
\begin{equation}
{\bf w}(n+1)\,=\,{\bf w}(n) + \frac{\mu {\bf G}(n){\bf
x}(n)e(n)}{ {\bf x}^T(n){\bf G}(n){\bf x}(n)+ \delta_P},
\end{equation}
where ${\bf x}(n)\,=\,[x(n),x(n-1),\cdots, x(n-L+1)]^T$ is the
input regressor vector, ${\bf G}(n)$ is a diagonal matrix that
modifies the step size of each tap, $\mu $ is the overall step
size, $\delta_P $ is a regularization parameter and $e(n)
\;=\;d(n)\,-\,{\bf w}^T(n){\bf x}(n)$ is the filter output error,
with $d(n)$ denoting the so-called desired response. In the system
identification problem under consideration, $d(n)$ is the observed
system output, given as ${d}(n)\;=\;{\bf w}_{opt}^T\,{\bf
x}(n)+v(n)$, where ${\bf w}_{opt}$ is the system impulse response
vector (supposed to be sparse), $x(n)$ is the system input and
$v(n)$ is an observation noise which is assumed to be white with
variance $\sigma_{v}^{2}$ and independent of $x(m)$ for all $n$
and $m$.

The matrix ${\bf G}(n)$ is evaluated as
\begin{equation}
{\bf G}(n)\;=\;diag(g_0(n),g_1(n)....g_{L-1}(n)),
\end{equation}
where,
\begin{equation}
g_l(n)\,=\,\frac{\gamma_l(n)}{\sum_{\substack{i=0}}^{\substack{L-1}}\gamma_i(n)},\,0\leq{l}\leq{L-1},
\end{equation}
with
\begin{eqnarray}
\gamma_l(n)=\max[\rho_g\max[\delta,\mid{w_{0}(n)}\mid,..\mid{w_{L-1}(n)}\mid], \nonumber\\
\mid{w_{l}(n)}\mid].
\end{eqnarray}
The parameter $\delta$ is an initialization parameter that helps
to prevent stalling of the weight updating at the initial stage
when all the taps are initialized to zero. Similarly, if an
individual tap weight becomes very small, to avoid stalling of the
corresponding weight update recursion, the respective
$\gamma_l(n)$ is taken as a small fraction (given by the constant
$\rho_g$) of the largest tap magnitude. By providing separate
effective step size $\mu\,g_l(n)$ to each $l$-th tap where
$g_l(n)$ is broadly proportional to $|w_l(n)|$, the PNLMS
algorithm achieves higher rate of convergence initially,  caused
primarily by the active taps. At a later stage, however, the
convergence slows down considerably, being controlled primarily by
the numerically dominant inactive taps that have progressively
diminishing effective step sizes \cite{IPNLMS},\cite{Deng}.

It has recently been shown \cite{yukawaProjection} that the PNLMS
weight update recursion (i.e., Eq. (1)) can be obtained by
minimizing the cost function $\norm{{\bf w}(n+1)-{\bf
w}(n)}^2_{{\bf G}^{-1}(n)}$ subject to the condition $d(n)-{\bf
w}^T(n + 1){\bf x}(n) = 0$ (the notation $\parallel{\bf
x}\parallel_{\bf A}^2$ indicates the generalized inner product
${\bf x}^T{\bf A}{\bf x}$)). In order to derive the ZA-PNLMS
algorithm, following \cite{Chen}, we add an $l_1$ norm penalty
$\gamma\norm{{\bf G}^{-1}(n){\bf w}(n+1)}_1$ to the above cost
function, where $\gamma$ is a very very small constant. Note that
unlike \cite{Chen}, we have, however, used a generalized form of
$l_1$ norm penalty here which scales the elements of ${\bf
w}(n+1)$ by ${\bf G}^{-1}(n)$ first before taking the $l_1$ norm
(the above scaling makes the $l_1$ norm penalty governed primarily
by the inactive taps). The above constrained optimization problem
may then be stated as,:
\begin{equation}
\smash{\displaystyle\min_{{\bf w}(n+1)}}\parallel{\bf w}(n+1)-{\bf
w}(n)
 \parallel^2_{{ \bf G}^{-1}} + \gamma\parallel{ \bf G}^{-1}{\bf
 w}(n+1)\parallel_1
 \end{equation}
subject to $d(n)-{\bf w}^T(n+1){\bf x}(n)~=~0$, where the short
form notation ``${\bf G}^{-1}$" is used to indicate ${\bf
G}^{-1}(n)$.
 Using Lagrange multiplier $\lambda$, this amounts to minimizing
 the cost function
$J(n+1)=\parallel{\bf w}(n+1)-{\bf w}(n)\parallel^2_{{ \bf
G}^{-1}}+ \gamma\parallel{ \bf G}^{-1}{\bf w}(n+1)\parallel_1
 +\lambda(d(n)-{\bf w}^T(n+1){\bf x}(n))$. Setting ${\partial J(n+1)}/{\partial {\bf w}(n+1)}=0$, one
 obtains,
 \begin{equation}
  {\bf w}(n+1)\,=\,{\bf w}(n)-[\gamma sgn({\bf w}(n+1)) - \lambda{ \bf G}(n){\bf x}(n)]
 \end{equation}
 where  $sgn(.)$ is the well known signum function, i.e.,
$sgn(x)=1\;(x>0),\;0\;(x=0),\;-1\;(x<0)$. Premultiplying both the
LHS and the RHS of (6) by ${\bf x}^T(n)$ and using the condition
$d(n)-{\bf w}^T(n+1){\bf x}(n)~=~0$, one obtains,
 \begin{equation}
  \lambda=\frac{e(n)+\gamma{\bf x}^T(n)sgn({\bf w}(n+1))}{{\bf x}^T(n){ \bf G}(n){\bf x}(n)}.
 \end{equation}
 Substituting (7) in (6), we have,
 \begin{eqnarray}
 &&{\bf w}(n+1)={\bf w}(n)+\frac{e(n){ \bf G}(n){\bf x}(n)}{{\bf x}^T(n){ \bf G}(n){\bf x}(n)}\nonumber\\
 &&-\gamma\left[I-\frac{{\bf x}(n){\bf x}^T(n){ \bf G}(n)}{{\bf x}^T(n){ \bf G}(n){\bf x}(n)}\right]sgn({\bf w}(n+1)).
 \end{eqnarray}
 Note that the above equation does not provide the desired weight update
 relation, as the R.H.S. contains the unknown term $sgn({\bf
 w}(n+1))$. In order to obtain a feasible weight update equation,
 we approximate $sgn({\bf w}(n+1))$ by an estimate, namely,
 $sgn({\bf w}(n))$ which is known. This is based on the assumption
 that most of the weights do not undergo change of sign as they get
 updated. This assumption may not, however, appear to be a very
 accurate one, especially for the inactive taps that fluctuate
 around zero value in the steady state. Nevertheless, an analysis
 of the proposed algorithm, as given later in this paper, shows that
 this approximation does not have any serious effect on the convergence behavior of the
 proposed algorithm. Apart from this, we also observe that in (8),
 elements of the matrix $\frac{{\bf
 x}(n){\bf x}^T(n){ \bf G}(n)}{{\bf x}^T(n)
  { \bf G}(n){\bf x}(n)}$ have magnitudes much less than 1,
  especially for large order filters, and thus, this term can be neglected in comparison
  to $\bf I$.

 From above and introducing the
  algorithm step size $\mu$ and a regularization parameter $\delta_P$ in (8),
  for a large order adaptive filter, one then obtains the following weight update equation :
\begin{equation}\label{ZA_PNLMS_update}
{\bf w}(n+1)={\bf w}(n)+\frac{\mu e(n){ \bf G}(n){\bf x}(n)}{{\bf
x}^T(n) {\bf G}(n){\bf x}(n)+\delta_P}-\rho sgn({\bf w}(n))
\end{equation}
where $\rho=\mu\gamma$.

Eq. (9) provides the weight update relation for the proposed
ZA-PNLMS algorithm, where the second term on the R.H.S. is the
usual PNLMS update term while the last term, i.e., $\rho sgn({\bf
w}(n))$ is the so-called zero attractor. The zero attractor adds
$-\rho sgn(w_j(n))$ to $w_j(n)$ and thus helps in its shrinkage to
zero. Ideally, the zero attraction should be confined only to the
inactive taps, which means that the proposed ZA-PNLMS algorithm
will perform particularly well for systems that are highly sparse,
but its performance may degrade as the number of active taps
increases. In such cases, Eq. (9) may be further refined by
applying the reweighting concept \cite{Chen} to it. For this, we
replace the $l_1$ regularization term $\parallel{ \bf G}^{-1}{\bf
w}(n+1)\parallel_1$ in (5) by a log-sum penalty
$\sum_{\substack{i=1}}^{\substack{L}}\frac{1}{g_i(n)}log(1+\mid
w_i(n+1)\mid/\epsilon)$ where $g_i(n)$ is the $i$-th diagonal
element of ${\bf G}(n)$ and $\epsilon$ is a small constant.
Following the same steps as used above to derive the ZA-PNLMS
algorithm, one can then obtain the RZA-PNLMS weight update
equation as given by
\begin{eqnarray}
{w}_i(n+1)&=&{w}_i(n)+\frac{\mu g_i(n){x}(n-i+1)e(n)}{
{\bf x}^T(n){ \bf G}(n){\bf x}(n)+ \delta_P}\nonumber\\
 &-&\rho \frac{  sgn({w}_i(n))}{1+\varepsilon\mid{w_i(n)}\mid},
\;i=0,\,1,\cdots, L-1,\nonumber\\
\end{eqnarray}
where $\varepsilon={1}/{\epsilon}$ and $\rho=\mu\gamma\varepsilon$.
The last term of (10), named as reweighted zero attractor,
provides a selective shrinkage to the taps. Due to this reweighted
zero attractor, the inactive taps with zero magnitudes or
magnitudes comparable to $1/\varepsilon$ undergo higher shrinkage
compared to the active taps which enhances the performance both in
terms of convergence speed and steady state EMSE.
\section{Convergence Analysis of the Proposed ZA-PNLMS Algorithm}
A convergence analysis of the PNLMS algorithm is known to be a
daunting task, due to the presence of ${\bf G}(n)$ both in the
numerator and the denominator of the weight update term in (1),
which again depends on ${\bf w}(n)$. The presence of the zero
attractor term makes it further complicated for the proposed
ZA-PNLMS algorithm, i.e., Eq. (9). To analyze the latter, we
follow here an approach adopted recently in \cite{Rajib_PNLMS} in
the context of PNLMS algorithm. This involves development of an
equivalent transform domain model of the proposed algorithm first.
A convergence analysis of the proposed algorithm is then carried
out by applying to the equivalent model a scheme of angular
discretization of continuous valued random vectors proposed first
by Slock \cite{Slock} and used later by several other researchers
\cite{Sankaran}, \cite{Paul}.
\subsection{A Transform Domain Model of the Proposed Algorithm}
The proposed equivalent model uses a diagonal `transform' matrix
${\bf G}^{\frac{1}{2}}(n)$ with $[{\bf
G}^{\frac{1}{2}}(n)]_{i,i}=g^{\frac{1}{2}}_i(n),
i=0,1,\cdots,L-1$, to transform the input vector ${\bf x}(n)$ and
the filter coefficient vector ${\bf w}(n)$ to their `transformed'
versions, given respectively as ${\bf s}(n)\;=\;{\bf
G}^{\frac{1}{2}}(n){\bf x}(n)$ and ${\bf w}_{N}(n)=[{\bf
G}^{\frac{1}{2}}(n)]^{-1}{\bf w}(n)$. It is easy to check that
${\bf w}^T_{N}(n){\bf s}(n)={\bf w}^T(n){\bf x}(n)\equiv{\bf
y}(n)$ (say), i.e., the filter ${\bf w}_{N}(n)$ with input vector
${\bf s}(n)$ produces the same output $y(n)$ as produced by ${\bf
w}(n)$ with input vector ${\bf x}(n)$. To compute ${\bf
G}^{\frac{1}{2}}(n+1)$ and ${\bf w}_{N}(n+1)$, the filter ${\bf
w}_{N}(n)$ is first updated to a weight vector ${\bf
w}^{'}_{N}(n+1)$ as
\begin{eqnarray}\label{ZA_PNLMS_TR}
 {\bf w}^{'}_N(n+1)&=&{\bf w}_N(n)+\frac{\mu e(n){\bf s}(n)}{{\bf s}^T(n){\bf s}(n)+\delta_P}\nonumber\\
&-& \rho {\bf G}^{-\frac{1}{2}}(n)sgn({\bf w}_N(n)).
\end{eqnarray}
 From (9), it is easy to check that ${\bf w}(n+1)$ is given by ${\bf
w}(n+1)={\bf G}^{\frac{1}{2}}(n){\bf w}^{'}_{N}(n+1)$. The matrix
${\bf G}(n+1)$ follows from ${\bf w}(n+1)$ following its
definition and ${\bf w}_{N}(n+1)$ is then evaluated as ${\bf
w}_{N}(n+1)=[{\bf G}^{\frac{1}{2}}(n+1)]^{-1}{\bf w}(n+1)$.
From above, it follows that ${\bf w}_N(n+1)\;=\;{\bf
G}^{-\frac{1}{2}}(n+1){\bf w}(n+1)\;=\;{\bf
G}^{-\frac{1}{2}}(n+1){\bf G}^{\frac{1}{2}}(n){\bf w}^{'}_N(n+1)$,
meaning $[{\bf
w}_N(n+1)]_i~=~[\frac{g_i(n)}{g_i(n+1)}]^{\frac{1}{2}}[{\bf
w}^{'}_N(n+1)]_i, i=0,1,\cdots, L-1.$ Since
$\sum^{L-1}_{i=0}g_i(n)=1$ and $0<g_i(n)<1, i=0,1,\cdots,L-1$, it
is reasonable to expect that $g_i(n)$ does not change
significantly from
 index $n$ to index $(n+1)$ [especially near convergence and/or for large order filters] and thus, we can make the approximation
$[g_i(n)]^{\frac{1}{2}}[{\bf
w}^{'}_N(n+1)]_i\approx[g_i(n+1)]^{\frac{1}{2}}[{\bf
w}^{'}_N(n+1)]_i$, which implies ${\bf w}^{'}_N(n+1)~=~{\bf
w}_N(n+1)$.

\subsection{Angular Discretization of a Continuous Valued Random Vector \cite{Slock}}
As per this, given a zero mean, $L\times 1$ random vector ${\bf
x}$ with correlation matrix ${\bf R}=E[{\bf x}{\bf x}^T]$, it is
assumed that ${\bf x}$ can assume only one of the $2L$ orthogonal
directions, given by $\pm{\bf e}_i, i=0,1,\cdots,L-1$, where ${\bf
e}_i$ is the $i$-th normalized eigenvector of ${\bf R}$
corresponding to the eigenvalue $\lambda_i$. In particular, ${\bf
x}$ is modeled as ${\bf x}\,=\,s\,r\,{\bf v}$, where ${\bf v} \in
\{{\bf e}_i| i=0,1,\cdots,L-1\}$, with probability of ${\bf
v}={\bf e}_i$ given by $p_i$, $r=\norm{\bf x}$, i.e., $r$ has the
same distribution as that of $\norm{\bf x}$ and $s \in
\{1,\,-1\}$, with probability of $s=\pm 1$ given by $P(s=\pm
1)=\frac{1}{2}$. Further, the three elements $s,\,r$ and ${\bf v}$
are assumed to be mutually independent. Note that as $s$ is zero
mean, $E[s\,r\,{\bf v}]={\bf 0}$ and thus $E[{\bf x}]={\bf 0}$ is
satisfied trivially. To satisfy $E[{\bf x}{\bf x}^T]={\bf R}$, the
discrete probability $p_i$ is taken as
$p_i=\frac{\lambda_i}{\textit{Tr}\,[{\bf R}]}$, which satisfies
$p_i\ge 0$, $\sum_{i=0}^{L-1}p_i=1$ and leads to $E[{\bf x}{\bf
x}^T]=E(s^2\,r^2\,{\bf v}{\bf v}^T)=E(r^2)\,E({\bf v}{\bf
v}^T)=\textit{Tr}\,[{\bf R}]\sum_{i=0}^{L-1}p_i{\bf e}_i{\bf
e}^T_i=\sum_{i=0}^{L-1}\lambda_i{\bf e}_i{\bf e}^T_i={\bf R}$.
Also note that if $\theta_i$ be the angle between ${\bf x}$ and
${\bf e}_i$, then $\cos(\theta_i)=\frac{{\bf x}^T{\bf e}_i}{||{\bf
x}||}$ and $E[\cos^2(\theta_i)]\approx
\frac{\lambda_i}{\textit{Tr}\,[{\bf R}]}$, meaning $p_i$ provides
a measure of how far ${\bf x}$ is (angularly) from ${\bf e}_i$ on
an average.

In our analysis of the proposed algorithm, we use the above model
to represent the transformed input vector ${\bf s}(n)$ as
\begin{equation}\label{slock_regressor}
 {\bf s}(n)=s_s(n)\,r_s(n)\,{\bf v}_s(n),
\end{equation}
where, $s_s(n)\in \{+1,-1\}$ with $P(s_s(n)=\pm1)=\frac{1}{2}$,
$r_s(n)=\norm{{\bf s}(n)}$ and ${\bf v}_s(n)\in \{{\bf
e}_{s,i}(n)|\, i=0,1,\cdots,L-1\}$ with $P({\bf v}_{s}(n)={\bf
e}_{s,i}(n))=\frac{\lambda_{s,i}(n)}{Tr({\bf S}(n))}$, where,
${\bf S}(n)=E[{\bf s}(n){\bf s}^T(n)]$, $\lambda_{s,i}(n)$ is the
$i$-th eigenvalue of ${\bf S}(n)$, and
 as before, the three elements $s_s(n),\,r_s(n)$ and ${\bf v}_s(n)$ are mutually
independent.

\subsection{Convergence of the ZA-PNLMS Algorithm in mean}

Now, defining the weight error vector at the $n$-th index as $\tilde{{\bf w}}(n)\,=\,{\bf w}_{opt}-{\bf w}(n)$, the transform domain weight error vector
 $\tilde{{\bf w}_N}(n)\,\,=\,\,{\bf G}^{-\frac{1}{2}}(n)\tilde{{\bf w}}(n)\,\equiv\,$ ${\bf G}^{-\frac{1}{2}}(n){\bf w}_{opt}-{\bf w}_N(n)$ and
expressing $e(n)={\bf s}^T(n)\tilde{{\bf w}}_N(n)+v(n)$, the recursive form of the weight error vectors can then be obtained as
\begin{eqnarray}\label{ZA_trans_wd}
\tilde{{\bf w}}_N(n+1)
&\approx&\tilde{{\bf w}}_N(n)-\frac{\mu {\bf s}(n){\bf s}^T(n)\tilde{{\bf w}}_{N}(n)}{ {\bf s}^T(n){\bf s}(n)+ \delta_P}\nonumber\\
&-&\frac{\mu {\bf s}(n)v(n)}{ {\bf s}^T(n){\bf s}(n)+ \delta_P}+ \rho {\bf G}^{-\frac{1}{2}}(n)sgn({\bf w}_N(n)).\nonumber\\
\end{eqnarray}
For our analysis here, we approximate $\delta_P$ by zero in (\ref{ZA_trans_wd})
as $\delta_P$ is a very very small constant.
The first order convergence
of the ZA-PNLMS is then provided in the following theorem.
\begin{theorem}\label{ch3_1st_order}
With a zero-mean  input $x(n)$ of covariance matrix ${\bf R}$, the ZA-PNLMS algorithm produces stable $\bar{\bf w}_N(n)$ and  also $\bar{\bf w}(n)$ if the step-size $\mu$ satisfies
 $0< \mu <2$ and under this condition, $\bar{\bf w}_N(n)$ and $\bar{\bf w}(n)$ converge respectively  as per the following:
\begin{eqnarray}\label{ZA_first_order_wn_final}
 & & \bar{\bf w}_N(\infty)={\displaystyle\lim_{n\rightarrow\infty}}\bar{\bf w}_N(n)=E({\bf G}^{-\frac{1}{2}}(n))\big|_{\infty}{\bf w}_{opt}\nonumber\\
&-&\frac{\rho}{\mu}Tr({\bf S}(\infty)){\bf S}^{-1}(\infty){\displaystyle\lim_{n\rightarrow\infty}}E({\bf G}^{-\frac{1}{2}}(n)
 sgn({\bf w}_N(n)))\nonumber\\
\end{eqnarray}
and
\begin{eqnarray}\label{ZA_first_order_w_final}
 & &\bar{\bf w}(\infty)={\displaystyle\lim_{n\rightarrow\infty}}\bar{\bf w}(n)={\bf w}_{opt}-\frac{\rho}{\mu}Tr({\bf S}(\infty))\times\nonumber\\
  & &E({\bf G}^{\frac{1}{2}}(n))\big|_{\infty}{\bf S}^{-1}(\infty){\displaystyle\lim_{n\rightarrow\infty}}E({\bf G}^{-\frac{1}{2}}(n)
 sgn({\bf w}(n)),\nonumber\\
\end{eqnarray}
where ${\bf S}(n)=E({\bf s}(n){\bf s}^T(n))=E({\bf G}^{\frac{1}{2}}(n){\bf R}{\bf G}^{\frac{1}{2}}(n))$.
\end{theorem}
\begin{proof}
 For analysis, we now substitute $\delta=0$ in (\ref{ZA_trans_wd}) as $\delta$ is a very very small constant. Taking expectation of both sides of
(\ref{ZA_trans_wd}) and invoking the well known ``independence assumption'' that allows taking ${\bf w}_N(n)$ to be statistically independent of ${\bf s}(n)$,
we then obtain,
\begin{eqnarray}\label{ZA_trans_Ewd}
 &E(\tilde{{\bf w}_N}(n+1))~=~ E(\tilde{{\bf w}_N}(n))-\mu E \left(\frac{ {\bf s}(n){\bf s}^T(n) }{ {\bf s}^T(n){\bf s}(n)}\right)E(\tilde{{\bf w}_{N}}(n))+ \rho E({\bf G}^{-\frac{1}{2}}(n)sgn({\bf w}_N(n)))\nonumber\\
&\Rightarrow E(\tilde{{\bf w}_N}(n+1))~=~ ({\bf I}-\mu {\bf B}(n))E(\tilde{{\bf w}_N}(n))+ \rho E({\bf G}^{-\frac{1}{2}}(n)sgn({\bf w}_N(n)))
\end{eqnarray}
where
\begin{equation}
 {\bf B}(n)~=~ E \left(\frac{ {\bf s}(n){\bf s}^T(n) }{ {\bf s}^T(n){\bf s}(n)}\right).
\end{equation}
Note that ${\bf B}(n)$ is symmetric and therefore, one can have its eigendecomposition  ${\bf B}(n)={\bf E}(n){\bf D}(n){\bf E}^T(n)$ where
${\bf E}(n)=[{\bf e}_0(n)\,{\bf e}_1(n)\cdots {\bf e}_{L-1}(n)]$, ${\bf D}(n)=\textit{diag}\,[\lambda_0(n),$ $\lambda_1(n),\cdots,\lambda_{L-1}(n)]$, with
${\bf e}_i(n)$ and $\lambda_i(n)$, $i=0,1,\cdots,L-1$ denoting the $i$-th eigenvector and eigenvalue of ${\bf B}(n)$ respectively. The eigenvalues are real
and the eigenvectors ${\bf e}_i(n)$ are mutually orthonormal, meaning ${\bf E}(n)$ is unitary, i.e., ${\bf E}^T(n){\bf E}(n)={\bf E}(n){\bf E}^T(n)={\bf I}$. From
${\bf B}(n){\bf e}_i(n)=\lambda_i(n){\bf e}_i(n)$ and the fact that $\norm{{\bf e}_i(n)}^2=1$, it is easy to observe that
\begin{equation*}
 \lambda_i(n)={\bf e}^T_i(n){\bf B}(n){\bf e}_i(n)=E\left[\frac{\left[{\bf s}^T(n){\bf e}_i(n)\right]^2}{\norm{{\bf s}(n)}^2}\right].
\end{equation*}
Two observations can be made now:
\begin{enumerate}
 \item  $\lambda_i(n)>0$ [ theoretically, one can have   $\lambda_i(n)=0$ also, provided ${\bf s}^T(n){\bf e}_i(n)=0$, i.e., ${\bf s}(n)$ is orthogonal to
${\bf e}_i(n)$ in each trial, which is ruled out here].
\item From Cauchy-Schwarz inequality, $[{\bf s}^T(n){\bf e}_i(n)]^2\leq \norm{{\bf s}(n)}^2\norm{{\bf e}_i(n)}^2=\norm{{\bf s}(n)}^2$, meaning
$\lambda_i(n)\leq 1$.
\end{enumerate}
%
%
Pre-multiplying both sides of  (\ref{ZA_trans_Ewd}) by ${\bf E}^T(n)$, defining ${\bf u}(n)={\bf E}^T(n)E(\tilde{{\bf w}_N}(n))$,
${\bf v}(n+1)={\bf E}^T(n)E(\tilde{{\bf w}_N}(n+1))$, ${\bf z}(n)={\bf E}^T(n)E({\bf G}^{-\frac{1}{2}}(n)sgn({\bf w}_N(n)))$,
 substituting ${\bf B}(n)$ by ${\bf E}(n){\bf D}(n){\bf E}^T(n)$ and using the unitariness of ${\bf E}(n)$,
we have,
\begin{equation}\label{ZA_ortho_eqn}
 {\bf v}(n+1)=({\bf I}-\mu {\bf D}(n)){\bf u}(n)+\rho {\bf z}(n).
\end{equation}
Taking norm on both sides of (\ref{ZA_ortho_eqn}) and invoking triangle inequality property of norm, i.e., $\norm{{\bf a}+{\bf b}}\leq\norm{{\bf a}}+\norm{{\bf b}}$, we then
obtain,
\begin{equation}\label{ZA_ortho_eqn_norm}
 \norm{{\bf v}(n+1)}\leq\norm{({\bf I}-\mu {\bf D}(n)){\bf u}(n)}+\rho \norm{{\bf z}(n)}.
\end{equation}
Since ${\bf E}(n)$ is unitary, we have $\norm{{\bf v}(n+1)}=\norm{E(\tilde{{\bf w}_N}(n+1))}$, $\norm{{\bf u}(n)}=\norm{E(\tilde{{\bf w}_N}(n))}$ and \\
$\norm{{\bf z}(n+1)}=\norm{E\left[{\bf G}^{-\frac{1}{2}}(n)sgn({\bf w}_N(n))\right]}$. Using the fact that
$\{ E[g^{-\frac{1}{2}}_{ii}(n)sgn({\bf w}_N(n))]\}^2$ $\leq E(g^{-1}_{ii}(n))$ (i.e., using Cauchy Schwarz inequality and the fact that
$sgn^2(.)=1$), we can write $\norm{{\bf z}(n+1)}\leq \sqrt{\sum_{i=0}^{L-1}E(g^{-1}_{ii}(n))}=c(n)$ (say).
 Clearly $c(n)$ is finite, as
${\bf G}(n)$ is a diagonal matrix with only positive elements.
From (\ref{ZA_ortho_eqn_norm}), one can then write,
\begin{equation}\label{ZA_ortho_norm}
 \norm{E(\tilde{{\bf w}_N}(n+1))}\leq\sqrt{\sum_{i=0}^{L-1}(1-\mu \lambda_i(n))^2|u_i(n)|^2}+\rho c(n).
\end{equation}
We now select $\mu$ so that $|1-\mu \lambda_i(n)|<1$, or , equivalently, $-1<1-\mu \lambda_i(n)<1$, which leads to the following:
\begin{enumerate}
 \item $\mu\lambda_i(n)>0$, meaning $\mu>0$ as $\lambda_i(n)>0$ (as explained above).
\item $\mu<\frac{2}{\lambda_i(n)}$ (since $\lambda_i(n)\leq1$ or, equivalently $\frac{1}{\lambda_i(n)}\geq 1$, it will be sufficient to take $\mu<2$ for satisfying this inequality).
\end{enumerate}
Therefore, for $0<\mu<2$, we have $|1-\mu \lambda_i(n)|<1$, where $0<\lambda_i(n)\leq1$. Let $\norm{E(\tilde{{\bf w}_N}(n))}=\theta(n)$
and $k(n)=\textit{max}\,\{|(1-\mu \lambda_i(n))|,\,i=0,1,\cdots,L-1\}$, meaning $0\leq k(n)<1$. From
(\ref{ZA_ortho_norm}), one can then write,
\begin{equation}
 \theta(n+1)\leq k(n)\sqrt{\sum_{i=0}^{L-1}|u_i(n)|^2}+ \rho c(n)=k(n)\theta(n)+ \rho c(n).
\end{equation}
Proceeding recursively backwards till $n=0$,
\begin{equation}\label{ZA_dynamic_theta}
 \theta(n+1)\leq {\displaystyle \prod_{i=0}^{n}}k(n-i)\,\theta(0)+ \rho\left(c(n)+\sum_{l=0}^{n-1}\left(\prod_{i=0}^{l}k(n-i)\right)c(n-l-1)\right).
\end{equation}
Clearly, for $0\leq k(n)<1$, the first term of RHS of (\ref{ZA_dynamic_theta}) vanishes as $n$ approaches infinity. For the second term, $c(n)$ is a
bounded sequence, which, in steady state, can be taken to be time invariant, say $c$, as the variation of $g_{ii}(n)$ vs. $n$, $i=0,1,\cdots,L-1$ in
the steady state are negligible. Also note that $\prod_{i=0}^{l}k(n-i)$ is a decaying function of $l$ since $0\leq k(m)<1$ at any index $m$. From these,
 one can write
${\displaystyle\lim_{n\rightarrow\infty}}\norm{E(\tilde{{\bf w}_N}(n))}={\displaystyle\lim_{n\rightarrow\infty}}\theta(n)\leq \rho K$ where $K$ is a
positive constant.
Recalling that
$\tilde{{\bf w}}(n)={\bf G}^{\frac{1}{2}}(n)\tilde{{\bf w}_N}(n)$, we can then write
${\displaystyle\lim_{n\rightarrow\infty}}\norm{E(\tilde{{\bf w}}(n))}\approx{\displaystyle\lim_{n\rightarrow\infty}}
\norm{E({\bf G}^{\frac{1}{2}}(n))E(\tilde{{\bf w}_N}(n))}<{\displaystyle\lim_{n\rightarrow\infty}}\norm{E(\tilde{{\bf w}_N}(n))}\leq \rho K$. Since
$\rho$ is very small, this
implies that $E({\bf w}(n))$ will remain in close vicinity of ${\bf w}_{opt}$ in the steady state under the condition: $0<\mu<2$. In other words, $E({\bf w}(n))$
will provide a biased estimate of ${\bf w}_{opt}$, though the bias, being proportional to $\rho$, will be negligibly small.

Under the condition $0<\mu<2$, letting $n$ approach infinity on both the LHS and the RHS of
(\ref{ZA_trans_Ewd}) and noting that as $n\rightarrow\infty$, $E(\tilde{{\bf w}_N}(n+1))\approx E(\tilde{{\bf w}_N}(n))$,
one can obtain from (\ref{ZA_trans_Ewd}),
\begin{eqnarray}\label{ZA_first_order_wn}
  {\displaystyle\lim_{n\rightarrow\infty}}E(\tilde{{\bf w}}_N(n))= \frac{\rho}{\mu}{\bf B}^{-1}(\infty){\displaystyle\lim_{n\rightarrow\infty}}E({\bf G}^{-\frac{1}{2}}(n)
\times sgn({\bf G}^{\frac{1}{2}}(n){\bf w}_N(n))).
\end{eqnarray}
Further, ${\bf B}(n)$ can be simplified in terms of ${\bf S}(n)$ by invoking the angular discretization model of a
random vector as discussed in the section III.B. We replace ${\bf s}(n)$ by $s_s(n)r_s(n){\bf v}_s(n)$ as given by (\ref{slock_regressor}). One can then write,
\begin{eqnarray}\label{Bn_simplification}
 {\bf B}(n)&=&E\left(\frac{{\bf s}(n){\bf s}^T(n)}{{\bf s}^T(n){\bf s}(n)}\right)
=E\left(\frac{r^2_s(n){\bf v}_s(n){\bf v}_s^T(n)}{r^2_s(n){\bf v}^T_s(n){\bf v}_s(n)}\right)\hspace{10mm} (\textit{since $s_s^2(n)=1$})\nonumber\\
&=&\sum_{i=0}^{L-1}\frac{\lambda_{s,i}(n)}{Tr({\bf S}(n))} {\bf e}_{s,i}(n){\bf e}^T_{s,i}(n)
=\frac{{\bf S}(n)}{Tr({\bf S}(n))},
\end{eqnarray}
since ${\bf S}(n)=\sum_{i=0}^{L-1}\lambda_{s,i}(n){\bf e}_{s,i}(n){\bf e}^T_{s,i}(n)$.

Letting $n$ approach infinity in (\ref{Bn_simplification}) and substituting this in (\ref{ZA_first_order_wn}),
\begin{eqnarray}\label{ZA_first_order_wn_final1}
  {\displaystyle\lim_{n\rightarrow\infty}}E(\tilde{{\bf w}}_N(n))=\frac{\rho}{\mu}Tr({\bf S}(\infty)){\bf S}^{-1}(\infty){\displaystyle\lim_{n\rightarrow\infty}}E({\bf G}^{-\frac{1}{2}}(n)
\times sgn({\bf G}^{\frac{1}{2}}(n){\bf w}_N(n))).
\end{eqnarray}
Recalling $\tilde{{\bf w}}(n)={\bf G}^{\frac{1}{2}}(n)\tilde{{\bf w}_N}(n)$, from (\ref{ZA_first_order_wn_final1}) we have,
\begin{eqnarray}\label{ZA_first_order_w_final1}
  {\displaystyle\lim_{n\rightarrow\infty}}E(\tilde{{\bf w}}(n))&\approx& {\displaystyle\lim_{n\rightarrow\infty}}E({\bf G}^{\frac{1}{2}}(n))E(\tilde{{\bf w}}_N(n))\nonumber\\
&=& \frac{\rho}{\mu}Tr({\bf S}(\infty)) E({\bf G}^{\frac{1}{2}}(n))\big|_{\infty}{\bf S}^{-1}(\infty){\displaystyle\lim_{n\rightarrow\infty}}E({\bf G}^{-\frac{1}{2}}(n)
 sgn({\bf w}(n))).\nonumber\\
\end{eqnarray}
Further, we have $E({\bf w}_N(n))= E({\bf G}^{-\frac{1}{2}}(n)){\bf w}_{opt}- E(\tilde{{\bf w}}_N(n))$ and $E({\bf w}(n))={\bf w}_{opt}- E(\tilde{{\bf w}}(n))$, and
thus, with (\ref{ZA_first_order_wn_final1}) and (\ref{ZA_first_order_w_final1}), this completes the proof.

\end{proof}
\begin{corollary}
 For white input, $\bar{w}_i(\infty) (={\displaystyle\lim_{n\rightarrow\infty}}E(w_i(n)))$ for the $i$-th active tap (i.e., for which $w_{opt,i}\neq0$) is approximately given by
\begin{equation}
 \bar{w}_i(\infty)=w_{opt,i}-\frac{\rho}{\mu}\bar{g}^{-1}_i(\infty)sgn(w_{opt,i})
\end{equation}
where $\bar{g}_i(\infty)={\displaystyle\lim_{n\rightarrow\infty}}\bar{g}_i(n)$ and $\bar{g}_i(n)=[E({\bf G}(n))]_{i,i}$.
\end{corollary}
\begin{proof}
For white input with variance $\sigma^2_x$,
we have ${\bf R}=\sigma^2_x{\bf I}$, ${\bf S}(n)=\sigma^2_xE({\bf G}(n))$, $Tr({\bf S}(n))=\sigma^2_x$ and ${\bf S}^{-1}(n)=\frac{1}{\sigma^2_x}E({\bf G}(n))^{-1}$
and then, we can have a simplified expression of $\bar{\bf w}(\infty)$ as
\begin{equation}\label{Ew_equation}
\bar{\bf w}(\infty)\approx{\bf w}_{opt}-\frac{\rho}{\mu}{\displaystyle\lim_{n\rightarrow\infty}}E({\bf G}(n))^{-1} E(sgn({\bf w}(n)))
\end{equation}
where we have assumed that in the steady state as ${n\rightarrow\infty}$, ${\bf G}^{-\frac{1}{2}}(n)$ and $sgn({\bf w}(n))$ become statistically independent and
$E({\bf G}(n)^{-\frac{1}{2}})\approx E({\bf G}(n))^{-\frac{1}{2}}$, which is reasonable as in the steady state, variance of each individual
 $g_i(n), i=0,1,\cdots,L-1$ is quite small (i.e., it behaves almost like a constant). Now, for an active tap with significantly large magnitude
$w_{opt,i}$, it is reasonable to approximate $sgn({ w}_i(n))\approx sgn(w_{opt,i})$ under the assumption that in the steady state,
the variance of ${ w}_i(n)$, i.e., $E(({ w}_i(n)-w_{opt,i})^2)$ is small enough compared to the magnitude of $w_{opt,i}$. Then, with
$E(sgn({ w}_i(n)))\approx E(sgn(w_{opt,i}))=sgn(w_{opt,i})$ for an active tap in the steady state, the result follows trivially from (\ref{Ew_equation}).
\end{proof}
Corollary 1 shows that
\begin{numcases}{\bar{w}_i(\infty)=}
     w_{opt,i}-\frac{\rho}{\mu}\bar{g}^{-1}_i(\infty),\, \text{if}\, sgn(w_{opt,i})>0\nonumber\\
     w_{opt,i}+\frac{\rho}{\mu}\bar{g}^{-1}_i(\infty),\,  \text{if}\, sgn(w_{opt,i})<0, \nonumber
    \end{numcases}
which implies that  $\bar{w}_i(\infty)$ is always closer to the origin vis-a-vis $w_{opt,i}$. Further, the bias
(i.e., usually defined as $w_{opt,i}-\bar{w}_i(\infty)$)
 is also proportional
to $\bar{g_i}^{-1}(\infty)$, meaning active taps with comparatively smaller values will have larger bias and vice versa.

In the case of inactive taps, we have $ w_{opt,i}=0$. From (14) and for $\rho=0$ (i.e., no zero attraction), this implies $\bar{w}_i(\infty)=0$, i.e., the
tap estimates fluctuate around zero value. For $\rho>0$, the zero attractors come into play in the update equation (7) and act as an additional
force that tries to pull the coefficients to zero from either side. The effect of zero attractor is thus to confine the fluctuations in a small
band around zero. On an average, one can then take $E(sgn({w}_i(n)))\big|_{\infty}\approx0$, meaning, from (16), the inactive tap estimates will largely
be free of any bias.

\section{Numerical Simulations}
In this section, we investigate evolution  of $E({\bf w}(n))$ of
the proposed ZA-PNLMS algorithm with time via simulation studies
in the context of  sparse system identification. For this, we
considered a sparse system with impulse responses of length L=512
as shown in Fig. 1. The system has 37 active taps and is driven by
a zero mean, white input $x(n)$  of variance $\sigma^2_x=1$, with
the output observation noise $v(n)$ being taken to be zero mean,
white Gaussian with $\sigma^2_v=10^{-3}$. The proposed ZA-PNLMS
algorithm is used to identify the system, for which the step size
$\mu$, the zero attracting coefficient $\rho$ and the
regularization parameter (to avoid division by zero) are taken to
be 0.7, 0.0001 and 0.01 respectively, while $\rho_g$ and $\delta$
are chosen as  $0.01$ and $0.001$ respectively. The simulations
are carried out for a total of 25,000 iterations and for each tap
weight $w_i(n)$, the learning curve $E[w_i(n)]$ vs $n$ is
evaluated by averaging $w_i(n)$ over 30 experiments. For
demonstration here, we consider four representative learning
curves, for $i$=37, 55, 67, 1. (the corresponding $w_{opt,i}$
given by 0.9, 0.1,-0.05 and 0 respectively). These are shown in
Figs. 2-5 respectively where it is seen that for both the inactive
tap (i.e., $w_{opt,1}$) and the active tap with relatively large
magnitude (i.e., $w_{opt,37}(n)$), $E[w_i(n)]$ converges to its
optimum values of 0 and 0.9 respectively. On the other hand, for
$w_{opt,67}(n)$ and $w_{opt,55}(n)$, i.e., for active taps with
relatively less magnitudes, $E[w_i(n)]$ converges with reasonably
large bias. This validates our conjectures made in section III
(Corollary 1 and the subsequent analysis). To validate the same
further, the bias is calculated from the learning curves (in the
steady state) for all the taps and then plotted in Fig. 6 against
the magnitude of the optimum tap weights. Clearly, the bias
becomes negligible as the magnitude of the active tap increases.

\begin{figure}
\begin{center}
\includegraphics[width=150mm,height=110mm]{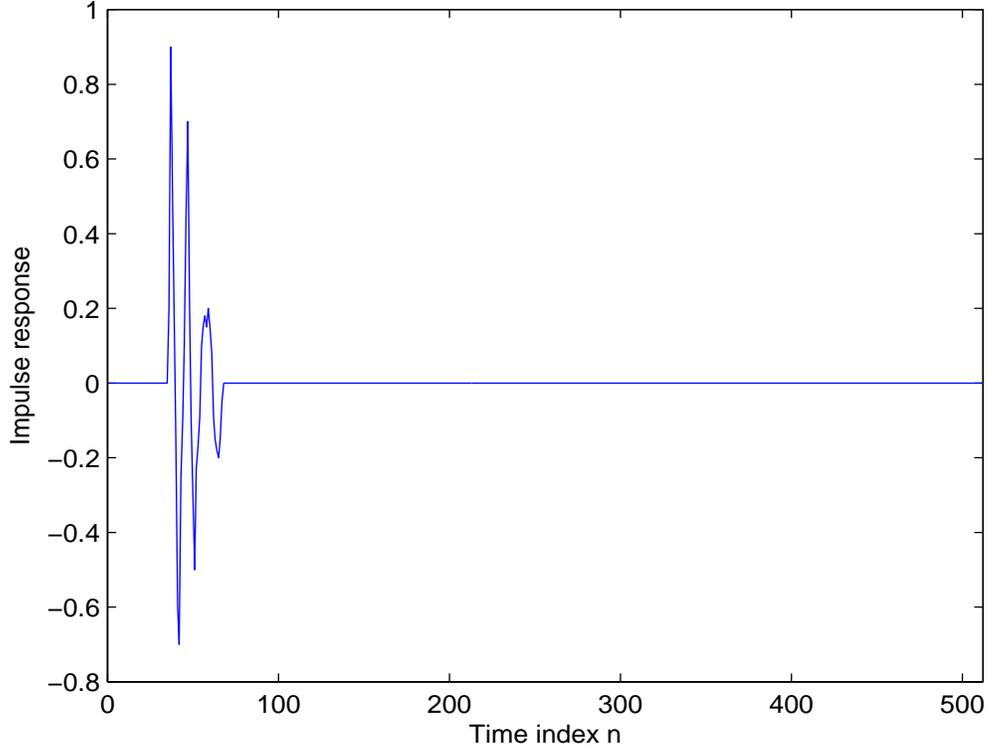}
\end{center}
\caption{Impulse response of the sparse system} 
\end{figure}
\begin{figure}
\begin{center}
\includegraphics[width=150mm,height=110mm]{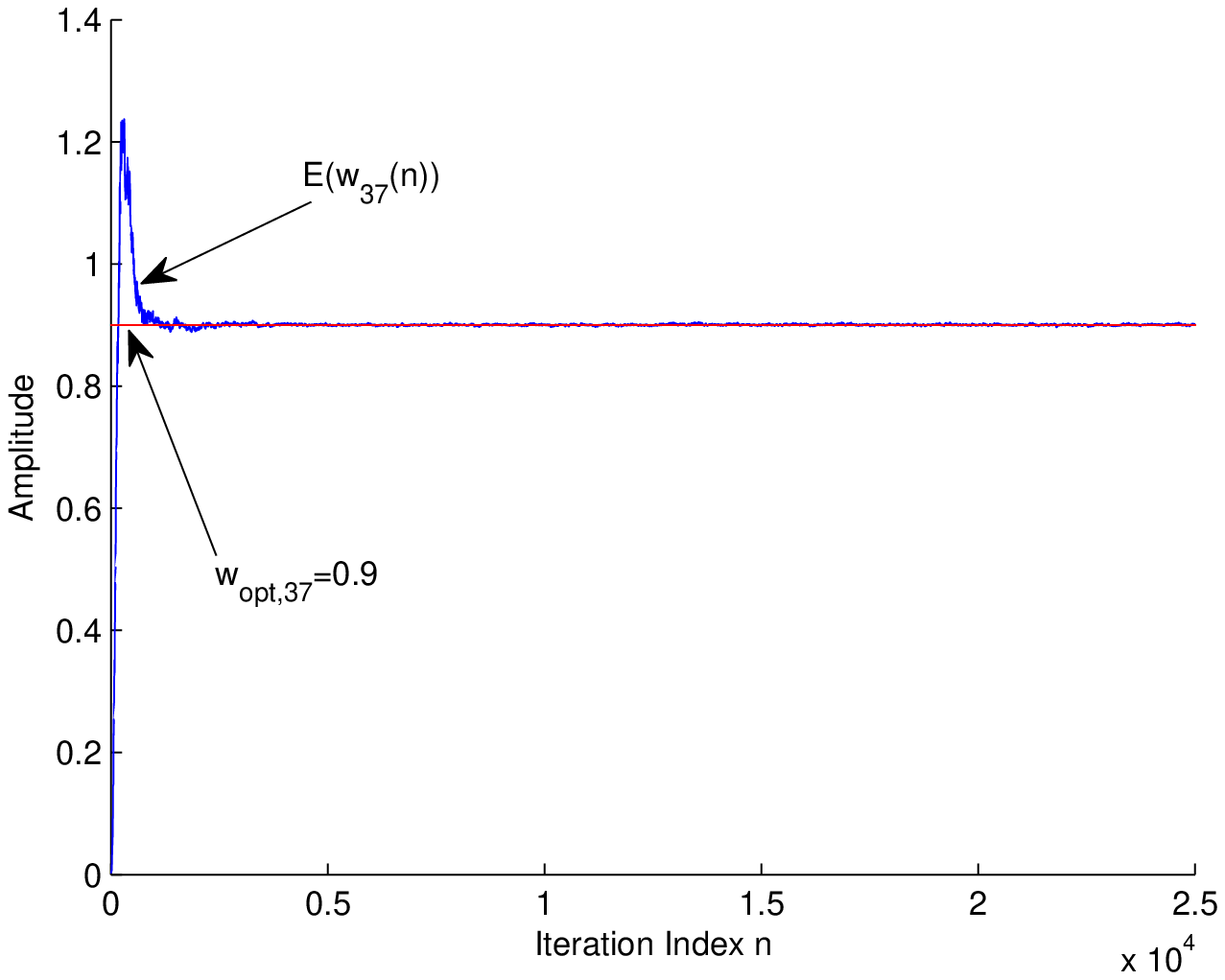}
\end{center}
\caption{Evolution of $E(w_{37}(n))$ of the ZA-PNLMS algorithm  with its optimum level $w_{opt,37}=0.9$.} 
\end{figure}
\begin{figure}
\begin{center}
\includegraphics[width=150mm,height=110mm]{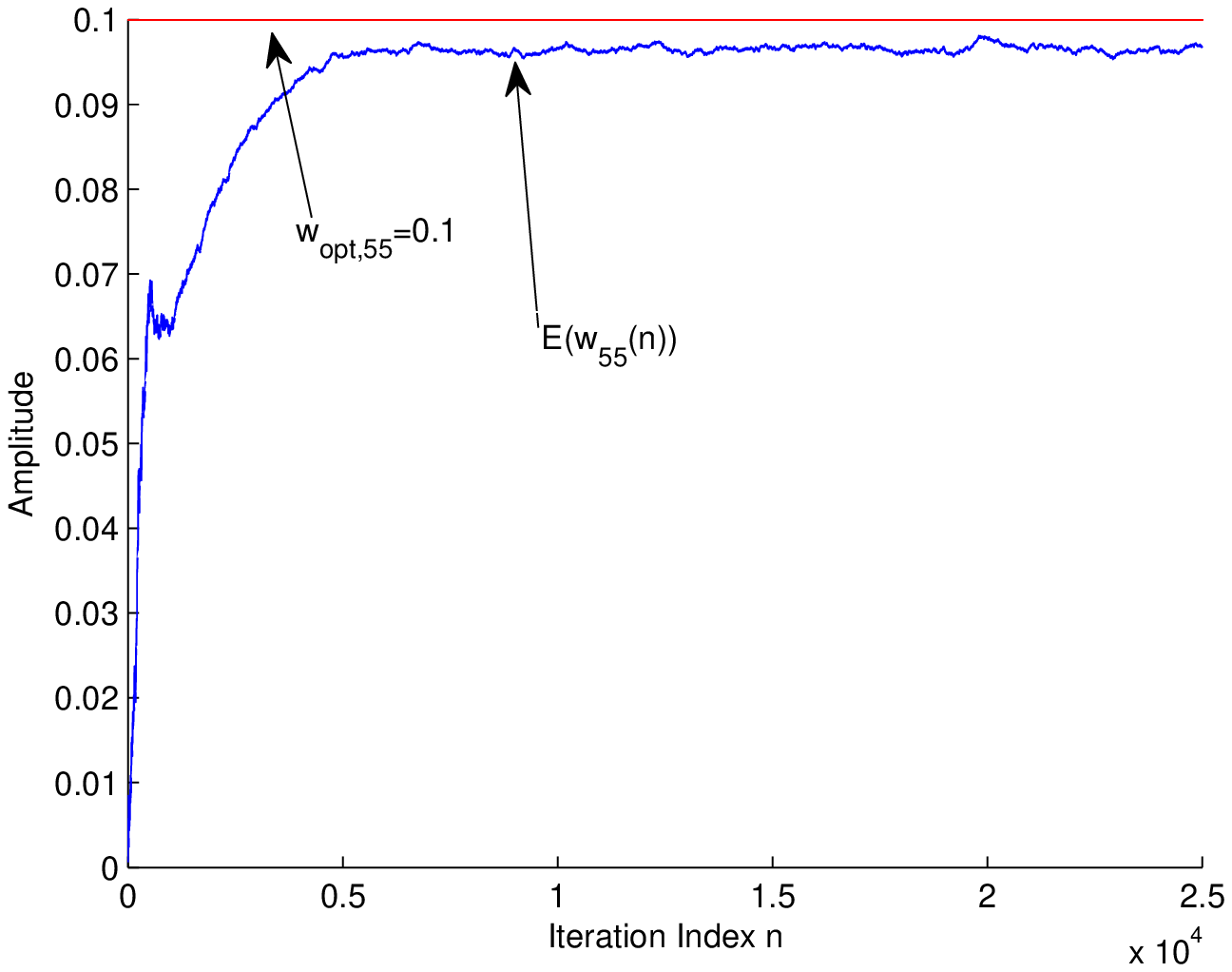}
\end{center}
\caption{Evolution of $E(w_{55}(n))$ of the ZA-PNLMS algorithm  with its optimum level $w_{opt,55}=0.1$.} 
\end{figure}
\begin{figure}
\begin{center}
\includegraphics[width=150mm,height=110mm]{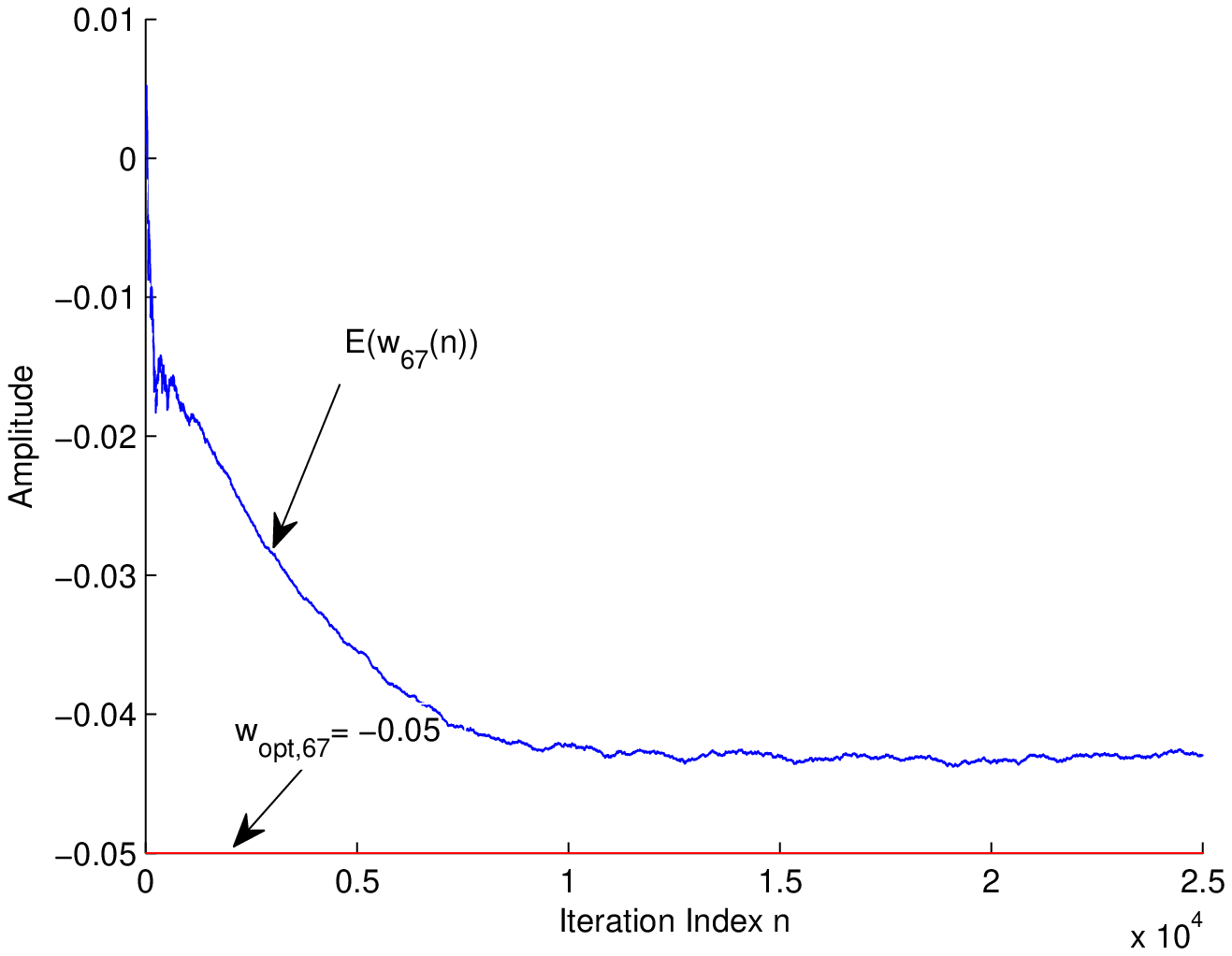}
\end{center}
\caption{Evolution of $E(w_{67}(n))$ of the ZA-PNLMS algorithm  with its optimum level $w_{opt,67}=-0.5$.} 
\end{figure}

\begin{figure}
\begin{center}
\includegraphics[width=150mm,height=110mm]{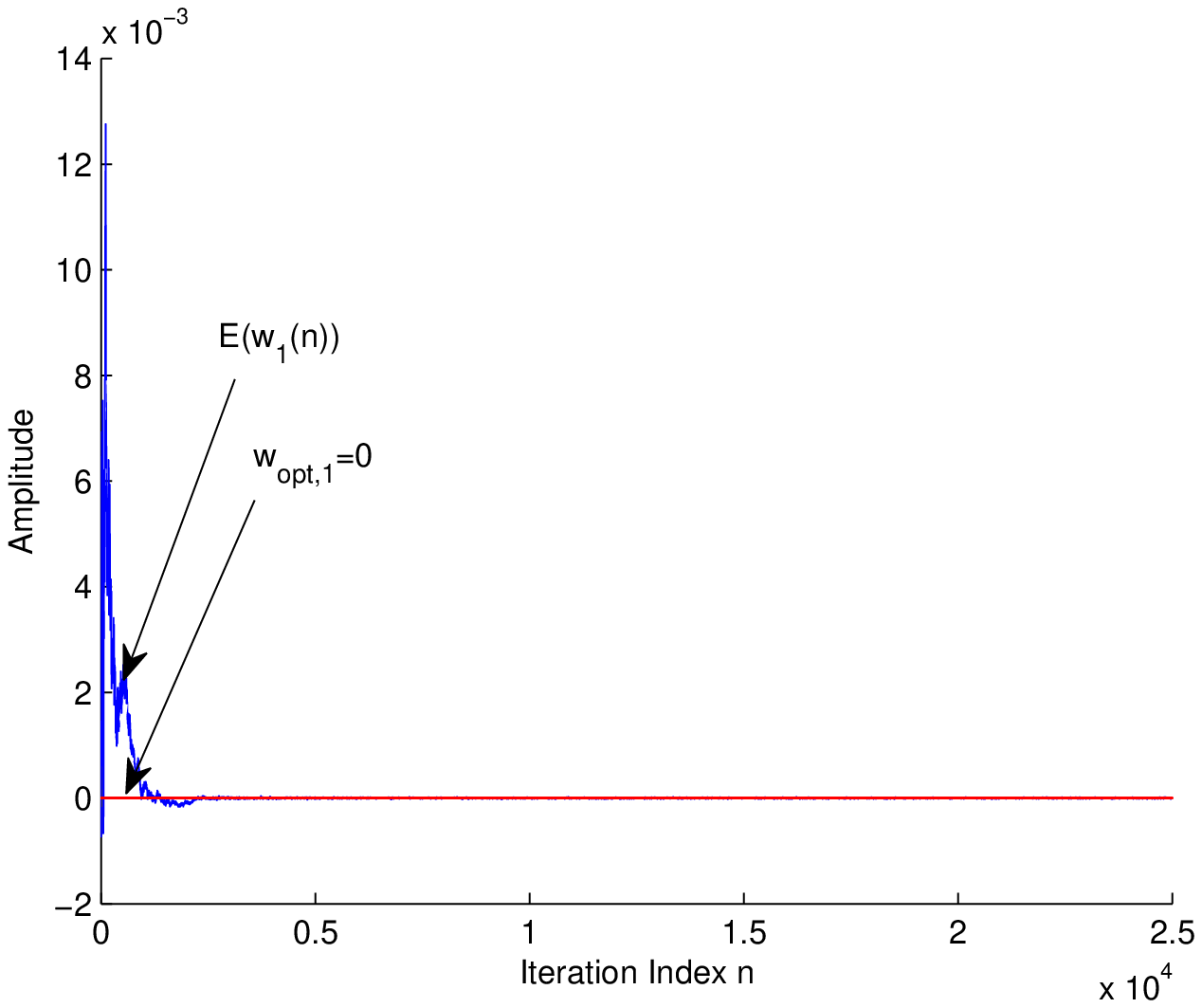}
\end{center}
\caption{Evolution of $E(w_{1}(n))$ of the ZA-PNLMS algorithm  with its optimum level $w_{opt,1}=0$.} 
\end{figure}

\begin{figure}
\begin{center}
\includegraphics[width=150mm,height=110mm]{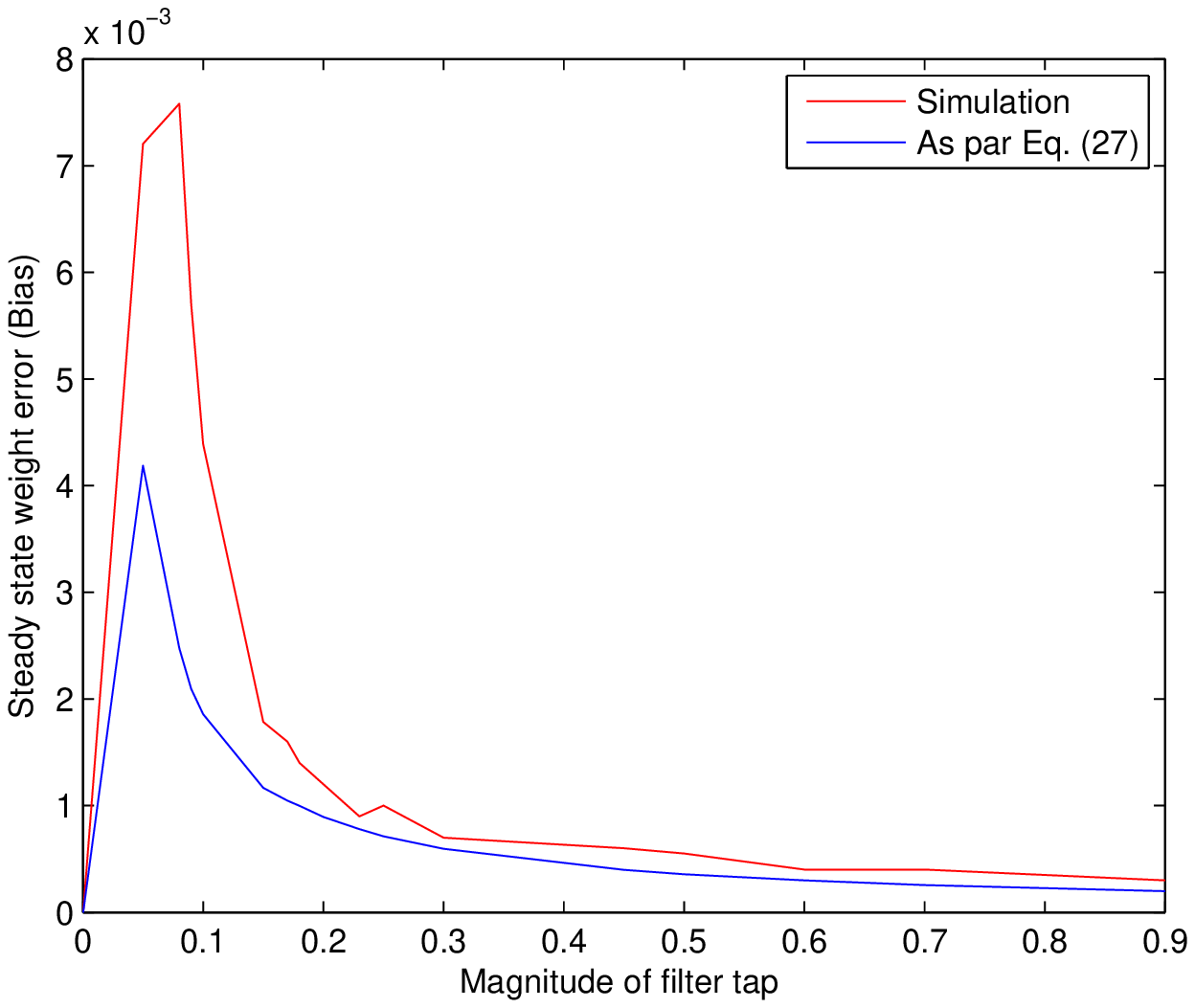}
\end{center}
\caption{Evolution of $E(w_{1}(n))$ of the ZA-PNLMS algorithm  with its optimum level $w_{opt,1}=0$.} 
\end{figure}

%
%



\end{document}